\documentclass[11pt]{article}

\usepackage{fullpage, amsmath, amssymb, amsthm, enumerate, color}
\usepackage{nicefrac}
\usepackage{algorithm}
\usepackage{multicol}
\usepackage[noend]{algpseudocode}
\usepackage{graphicx, tikz, algorithmicx, natbib}
\usepackage{dsfont}

\newtheorem{theorem}{Theorem}[section]
\newtheorem{proposition}[theorem]{Proposition}
\newtheorem{lemma}[theorem]{Lemma}
\newtheorem{claim}[theorem]{Claim}
\newtheorem{corollary}[theorem]{Corollary}
\newtheorem{definition}[theorem]{Definition}
\newtheorem{remark}[theorem]{Remark}
\usepackage{authblk}

\def\P{\mathbb{P}}

\def\E{\mathbb{E}}

\def\argmax{\text{argmax}}

\begin{document}

\title{Maximum Weight Online Matching with Deadlines\footnote{This paper is the result of the merger of three manuscripts: 
(i) Maximizing Efficiency in Dynamic Matching Markets
by Itai Ashlagi, Maximilien Burq, Patrick Jaillet, and Amin Saberi, available at https://arxiv.org/abs/1803.01285, (ii) Online Matching in a Ride-Sharing Platform
by Chinmoy Dutta and Chris Sholley, available at https://arxiv.org/abs/1806.10327, and (iii) Maximum Weight Matching in Random Permutation Model, by Itai Ashlagi, Maximilien Burq, and Amin Saberi, unpublished manuscript (2018). }}

\author[1]{Itai Ashlagi \thanks{iashlagi@stanford.edu}}
\author[2]{Maximilien Burq \thanks{mburq@mit.edu}}
\author[3]{Chinmoy Dutta \thanks{cdutta@lyft.com}}
\author[2]{Patrick Jaillet \thanks{jaillet@mit.edu}}
\author[1]{Amin Saberi \thanks{saberi@stanford.edu}}
\author[3]{Chris Sholley \thanks{chris@lyft.com}}
\affil[1]{MS\&E, Stanford University}
\affil[2]{Operations Research Center, MIT}
\affil[3]{Lyft}

\date{}
\maketitle
\begin{abstract}

We study the problem of matching agents who arrive at a marketplace over time and leave after $d$ time periods. Agents can only be matched while they are present in the marketplace. Each pair of agents can yield a different match value, and the planner's goal is to maximize the total value over a finite time horizon.

First we study the case in which vertices arrive in an adversarial order. We provide  a randomized $\nicefrac{1}{4}$-competitive algorithm building on a  result by \citet{feldman2009online} and \citet{lehmann2006combinatorial}. We extend the model to the case in which  departure times are drawn independently from a distribution with non-decreasing hazard rate, for which we establish a $\nicefrac{1}{8}$-competitive algorithm.

When the arrival order is chosen uniformly at random, we show that a batching algorithm, which computes a maximum-weighted matching every $(d+1)$ periods, is $0.279$-competitive.

%
\end{abstract}
\pagebreak

\section{Introduction}

Traffic congestion is a severe problem in metropolitan areas around the world. A resident in Los Angeles is estimated to lose around \$6,000  per year due to spending extra hours in traffic (Economist 2014). This does not account for  extra carbon emissions.  A couple of  ways to  relieve congestion are  pricing \citep{vickrey1965pricing}  and carpooling, and online platforms and other technological advances are now available to assist with these tasks \citep{ostrovsky2018carpooling}.\footnote{\cite{ostrovsky2018carpooling} discusses  the complementarities between autonomous vehicles, carpooling and pricing.} 

Online platforms now offer the option to share rides. An immediate benefit is that passengers who share rides pay a lower price for the trip. However, the passenger may also experience a disutility from additional waiting, detours, and less privacy. Facing these trade-offs, ride-sharing  platforms and carpooling applications seek to increase the volume of ride-sharing, which will in turn  help in reducing congestion. 

In this paper,  we present and study  a graph-theoretic matching problem that captures the following three key features faced by ride-sharing platforms. First is spatial; the farther away two passengers are from each other, the higher the disutility from being matched. Second is temporal;  passengers cannot match passengers who  request rides at very different times. Third, the platform faces uncertainty about  future demand.

\subsection*{Contributions}

Next we describe our basic graph-theoretic model and contributions. Time is discrete and one vertex of a given graph arrives at each time period. Every edge has a non-negative weight, representing the reward from matching these two vertices. A vertex cannot match more than $d$ periods after its arrival; after $d$ units of time the vertex becomes critical and departs. It is  helpful to think of $d$ as a service quality set by the platform and a passenger is assigned to a single ride after waiting for $d$ periods of time.

The goal is to find a weighted matching with a large total weight in an online manner. This means that the decision for every vertex has to be made no later than $d$ periods after its arrival (this differs from the classic online bipartite matching  literature, in which $d=0$). There is no a priori  information about weights  or arrival times and the  underlying graph may be arbitrary and hence {\it non-bipartite}.

Our first results are given in a setting, in which the vertices arrive in an {\bf adversarial order}.  We introduce for this setting a $\nicefrac{1}{4}$-competitive algorithm, termed {\it Postponed Greedy} (PG). We further show that no algorithm  achieves a competitive ratio that is higher than $\nicefrac{1}{2}$.

The key idea behind PG is  to look at a virtual bipartite graph, in which each vertex is duplicated into a ``buyer" and a ``seller" copy. We enforce that the seller copy does not match before the vertex becomes critical. This enables us to postpone the matching decision until we have more information about the graph structure and the likely matchings. We then proceed in a manner similar to \cite{feldman2009online}: tentatively match each new buyer copy to the seller that maximizes its margin, i.e., the difference between edge weight, and the value of the seller's current match. 

We extend the model to the case where the departure of vertices are determined stochastically. We show that when the departure distribution is memoryless and realized departure times are  revealed to the algorithm just as becoming critical, one can adapt the PG algorithm to achieve a competitive ratio of $\nicefrac{1}{8}$. It is worth noting that when departure times are chosen in an adversarial manner no algorithm can achieve a constant competitive ratio.

Next we study the setting in which vertices arrive in a {\bf random order}. We analyze  a {\it batching} algorithm which, every $d+1$ time steps, computes a maximum weighted matching among the last $d+1$ arrivals. Vertices that are left unmatched are discarded forever. We show that when the number of vertices is sufficiently large, batching is $0.279$-competitive. 

The analysis proceeds in three steps. First, we show that the competitive ratio is bounded by the solution to a graph covering problem. Second, we show how a solution for small graphs can be extended to covers for larger graphs. Finally, we establish a reduction that allows us to consider only a finite set of values for $d$. We conclude with a computer-aided argument for graphs in the finite family.


\subsection*{Related literature}

There is a growing literature related to  ride-sharing.  \citet{S14} finds that about $80\%$ of rides in Manhattan could be shared by two passengers. Many studies focus on  rebalancing or dispatching problems without pooling, e.g., \citet{PSFR12, ZP14, S14, SSGF16, banerjee2018state}.  \citet{ASWFR17} studies real-time high-capacity ride-sharing.  It does not consider, however, a graph-theoretic online  formulation for matching rides.

This paper is closely related to the  online matching literature.  In the classic problem, introduced in \citet{kvv}, the graph is bipartite with  vertices on one side waiting, while others are arriving sequentially and  have to be matched {\it immediately} upon arrival. This work has numerous extensions, for example to stochastic arrivals and in the  adwords context \cite{mehta2007adwords,aryanak_randominput,aryanak_stmatching,mos,STMatchingPatrick}. See \cite{mehta2013online} for a detailed survey. Our contributes to this literature in three ways. First, we provide algorithms that perform well on edge-weighted graphs. Second, our graph can be non-bipartite, which is the case in ride-sharing and kidney exchange. Third, all vertices can arrive over time and may remain for some given time until they are matched. Closely related is \citet{huang2018match}, which studies a similar model to ours in the non-weighted case, but allow departure times to be adversarial. 

Several  papers consider the problem of dynamic matching in the edge-weighted case. \citet{feldman2009online} find that in the classic online bipartite setting, no algorithm achieves a constant approximation. They introduce a  {\it free disposal} assumption, which allows to discard a matched vertex in favor of a new arriving vertex. They find, based on an algorithm by \citet{lehmann2006combinatorial}, that a  greedy algorithm that matches a vertex to the highest marginal vertex, is $0.5$-competitive. We build on this result for a special classes of bipartite graphs.   In the adversarial setting \citet{emek2016online,ashlagi2017min} study the problem of minimizing the sum of distances between matched vertices and the sum of their  waiting times. In their model no vertex leaves unmatched and our model does not account for vertices' waiting times. Few  papers  consider the stochastic environment \citep{baccara2015optimal,ozkan2016dynamic,hu2016dynamic}.  These papers find that some waiting before matching  is beneficial for improving efficiency. 

Related to our work are some papers on job or packet scheduling. Jobs arrive over online to a buffer, and reveal upon arrival the deadline by which they need to be scheduled. The algorithm  can schedule at most one job per time and the value of scheduling a job is independent from the time slot. Constant approximation algorithms are given by \citet{chin2006online} and \citet{li2005optimal}.

Finally, there is a growing  literature that focuses on dynamic matching motivated from  kidney exchange \citep{Utku,anderson2015dynamic,dickerson2013failure,ashlagi2017matching}. These  papers focus mostly  on random graphs with no weights. Closer to our paper is \cite{akbarpour2017thickness}, which finds that in a sparse random graph, knowledge about the departure  time of a vertex is beneficial and matching a vertex  only when it becomes  critical  performs well.  Our work differs from these papers in two ways: we consider the edge-weighted case, and, we make no assumption on the graph structure.


\section{Model}

Consider a weighted graph $G$ with  $n$ vertices indexed by $i=1,\ldots n$. Vertices arrive  sequentially over $n$ periods and let $\sigma(i)$ denote the arrival time of vertex $i$. Let $v_{ij}\geq 0$ denote the  weight on the undirected edge $(i,j)$ between vertices $i$ and $j$.

For vertices $i$ and $j$ with $\sigma(i) < \sigma(j)$, the weight $v_{ij}$ on the edge between $i$ and $j$ is observed only after  vertex $j$ has arrived. 

For $d \geq 1$, the {\bf online graph with deadline $d$}, denoted by $G_{d,\sigma}$, has the same vertices as $G$, and the edge between $i$ and $j$ in $G$ exists if an only if  $|\sigma(i)-\sigma(j)|\leq d$.    We say that $i$ becomes {\bf critical} at period $\sigma(i)+d$, at which time the online algorithm needs to either match it and collect the associated edge weight, or let it {\it depart} from the graph. 


We will consider two settings regarding how arrivals are generated. In the Adversarial Order (AO) setting, we assume that $\sigma(i) = i$. In the Random Order (RO) setting, we assume that $\sigma$ is sampled uniformly at random among all possible permutations $S_n$ of $[1,n]$.

The goal is to find an online algorithm that generates a matching with high total weight. More precisely, we seek to design a randomized online algorithm that obtains in expectation a high fraction of the expected maximum-weight of a matching over $G_{d, \sigma}$.


To illustrate a natural tradeoff, consider the example in  Figure \ref{fig:hard:basic} for $d=1$. At period $2$ the planner can either match vertices $1$ and $2$ or let vertex $1$ remain unmatched. This simple example shows that no deterministic algorithm can obtain a constant competitive ratio. Furthermore, no algorithm can achieve a competitive ratio higher than $\nicefrac{1}{2}$. 

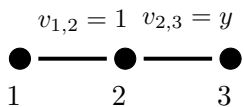
\begin{figure}[H]
  \centering
  \begin{tikzpicture}[pre/.style={<-,shorten <=1.5pt,>=stealth,thick}, post/.style={->,shorten >=1pt,>=stealth,thick}, scale=0.7]
  \tikzstyle{every node}=[draw,shape=rectangle,minimum size=5mm, inner sep=0];
  \tikzstyle{edge} = [draw,thick,-]
  \tikzstyle{every node}=[shape=circle,minimum size=8mm, inner sep=0];

  \draw [fill](-2,0) circle [radius=0.2];
  \node [right] at (-2.7,-0.7) {$1$};
  \draw [fill](0,0) circle [radius=0.2];
  \node [right] at (-0.7,-0.7) {$2$};
  \draw [fill](2,0) circle [radius=0.2];
  \node [right] at (1.3,-0.7) {$3$};

  \draw[line width=1.5pt] [-] (-2+0.35, 0) .. controls(-1,0) ..(0-0.35, 0);
  \node [right] at (-1.7,0.7) {\small $v_{1,2} = 1$};
  \draw[line width=1.5pt] [-] (0+0.35, 0) .. controls(1,0) .. (2-0.35, 0);
  \node [right] at (0.3,0.7) {\small $v_{2,3} = y$};
  \end{tikzpicture}
  \caption{Let $d = 1$. Therefore, there is no edge between vertices $1$ and $3$. The algorithm needs to decide whether to match $1$ with $2$ and collect $v_{1,2}$ without knowing $y$.}
  \label{fig:hard:basic}
\end{figure}


\section{Adversarial arrival order}
\label{sec:adverse}

The example in Figure \ref{fig:hard:basic} illustrates a necessary condition for the algorithm to achieve a constant competitive ratio: with some probability, vertex $2$ needs to forgo the match with vertex $1$. We  ensure this property  by assigning every vertex to be either a {\it seller} or a {\it buyer}. We then prevent sellers from matching before they become critical, while we allow buyers to be matched at any time.

It will be useful to first study a special case, in which the underlying graph $G$ is  bipartite,  with sellers on one side and buyers and in the online graph a buyer and a seller cannot match if the buyer arrives before the seller.  For such  online graphs we show that a greedy algorithm given by \citet{feldman2009online} is  0.5-competitive. We then build on this algorithm to design a randomized $\nicefrac{1}{4}$-competitive algorithm for arbitrary graphs.

\subsection{Bipartite constrained online graphs}
\label{sec:bipartite}

Let $G$ be a bipartite graph and $\sigma$ be the order of arrivals. The online graph $G_{d,\sigma}$ is called {\bf constrained bipartite} if for every seller  $s$ and every buyer $b$, there is no edge between $s$ and $b$ if $\sigma(b)<\sigma(s)$, i.e. $b$ and $s$ cannot match if $b$ arrives before $s$.

Consider the following greedy algorithm, which attempts to match buyers in their arriving order. An arriving buyer $b$ is matched to the seller with the highest marginal value if the marginal value is positive. If the seller is already matched to another buyer $b'$, $b'$ becomes unmatched and never gets matched again. Formally:
\begin{algorithm}
\caption{Greedy algorithm \citep{feldman2009online}}
\begin{itemize}
\item Input: constrained bipartite graph,  $G_{d,\sigma}$.
\item For each arrival $i=1,\ldots,n$:
\begin{itemize}
\item If $i$ is a seller, initialize $p(i)=0$, and $m(s) = \emptyset$.
\item If $i$ is a buyer:
\begin{itemize}
\item Set $s\in \argmax_{s'\in S}\{v_{is'}-p(s')\}$.
\item If $v_{is}-p(s)>0$,  set $m(s)=i$ and set $p(s)=v_{is}$.
\end{itemize}
\end{itemize}
\item When a seller $s$ becomes critical: match it to $b = m(s)$ if $m(s) \neq \emptyset$.
\end{itemize}
\end{algorithm}

\begin{proposition}[\citet{feldman2009online}]
\label{prop:factor2}
The greedy algorithm is 0.5-competitive for online bipartite constrained graphs.
\end{proposition}

\citet{feldman2009online} prove that  this algorithm is 0.5-competitive for an online matching problem with {\it free disposal}. In their setting all seller exists and buyer arrive one at a time. The algorithm provides the same guarantees for constrained bipartite graph since, by construction, there is no harm in assuming that all sellers exist rather than  arriving over time.  The key behind the proof is that the value $p(s)$ function  for each seller $s$ is  submodular. In fact the result is a special case of a result by \citet{lehmann2006combinatorial}, who study combinatorial auctions with submodular valuations.

\subsection{Arbitrary graphs}
\label{sec:arbitrary graphs}

In this section we extend the greedy algorithm for constrained bipartite graphs to arbitrary graphs.
A naive way to generate a online constrained bipartite graph from an arbitrary one is to randomly assign each vertex to be  either a seller or a buyer, independently and with probability \nicefrac{1}{2}. Then only keep the edges between each buyer and all the sellers who arrived before her. Formally:

\begin{algorithm}
\caption{Naive Greedy}
\begin{itemize}
\item Input: an online graph with deadline $d$, $G_{d,\sigma}$.
\item For each vertex $t=1,\ldots, n$:
\begin{enumerate}
\item[]  Toss a fair coin to decide whether $i$ is a  \emph{seller} or a \emph{buyer}. Construct the online constrained bipartite graph $\tilde{G}(d,\sigma)$  by keeping only the edges between each buyer and the sellers who arrived  before her.
\end{enumerate}
\item Run the Greedy algorithm on $\tilde{G}(d,\sigma)$.
\end{itemize}
\end{algorithm}

\begin{corollary}
  \label{cor:factor8}
  The naive greedy  algorithm is $\nicefrac{1}{8}$-competitive for  arbitrary online graphs.
\end{corollary}

Observe that for vertices $i$, $j$ with $\sigma(i) < \sigma(j)$, edge $(i,j)$ in the original graph remains in the generated constrained bipartite graph with probability $1/4$ (if $i$ is a seller and $j$ is a buyer). We then use proposition \ref{prop:factor2} to prove that  naive greedy is $\nicefrac{1}{8}$-competitive.

One source of inefficiency in the naive greedy algorithm is that the decision whether a vertex becomes  a seller or a buyer is done independently at random and without taking the graph structure into consideration. We next introduce  the   {\it Postponed Greedy} algorithm that defers these decisions  as long as possible in order to construct the constrained bipartite graph more carefully.

When a vertex $k$ arrives, we add two copies of $k$ to a virtual graph: a seller $s_k$ and a buyer $b_k$. Let $S_t$ and $B_t$ be the set of sellers and buyers at arrival time $t$. On arrival, seller $s_k$ does not have any edges, and buyer $b_k$ has edges towards any vertex $s_l \in S_k$ with value $v_{l,k}$. Then we run the greedy algorithm with the virtual graph as input. When  a vertex $k$ becomes critical, $s_k$ becomes critical in the virtual graph, and we compute its matches generated by greedy.

Both $s_k$ and $b_k$ can be matched in this process.  If we were to honor both matches, the outcome would correspond to a 2-matching, in which each vertex has degree at most 2. Now observe that because of the structure of the constrained bipartite graph, this 2-matching does not have any cycles; it is just a collection of disjoint paths. We decompose each path into two disjoint matchings and choose each matching with probability  $1/2$.

In order to do that, the algorithm must determine, for each original vertex $k$, whether the virtual buyer $b_k$ or virtual seller $s_k$ will be used in the final matching. We  say that $k$ is a \emph{buyer} or \emph{seller} depending on which copy is used. We say that vertex $k$ is \emph{undetermined} when the algorithm has not yet determined which virtual vertex will be used. When an undetermined vertex becomes critical, the algorithm  flips a fair coin to decide whether to match according to the buyer or seller copy. This decision is then propagated to the next vertex in the 2-matching: if $k$ is a \emph{seller} then the next vertex will be a \emph{buyer} and vice-versa.
That ensures that assignments are correlated and saves a factor $2$ compared to uncorrelated assignments in the naive greedy algorithm.

\begin{algorithm}[ht]
\caption{Postponed Greedy (PG)}
\label{alg:QDDA}

\begin{itemize}
\item Input: an online graph with deadline $d$, $G_{d,\sigma}$.


\item Process events at time $t$ in the following way:
\begin{enumerate}
  \item \emph{Arrival of a vertex $k$:}
  \begin{enumerate}
  	\item Set $k$'s status to be \emph{undetermined}.
    \item \emph{Add a virtual seller:} $S_t \leftarrow S_{t-1} \cup \{s_k\}$, $p(s_k) \leftarrow 0$ and $m(s_k) = \emptyset$.
    \item \emph{Add a virtual buyer:} $B_t \leftarrow B_{t-1} \cup \{b_k\}$.
    \item \emph{Find a virtual seller for the virtual seller:} $s = \argmax_{s' \in S_t} v_{s', b_k} - p(s')$. 
    \item \emph{Match if marginal utility is positive:} If $v_{s, b_k} - p(s) > 0$, then tentatively match $b_k$ to $s$ by setting $m(s) \leftarrow b_k$ and $p(s) \leftarrow v_{s, b_k}$.
  \end{enumerate}
  \item \emph{Vertex $k$ becomes critical:}
  \begin{enumerate}
  	\item \emph{Proceed if no match found:} If $m(s_k) = \emptyset$, do nothing.
    \item \emph{match in the virtual graph:} If $m(s_k)= b_l$.
Set $S_t \leftarrow S_t \setminus \{s_k\}$, and $B_t \leftarrow B_t \setminus \{b_l\}$. 
    \item If $k$'s status is \emph{undetermined}, w.p $\nicefrac{1}{2}$ set it to be either \emph{seller} or \emph{buyer}.
    \begin{enumerate}
      \item\label{step:seller}\emph{If $k$ is a seller:} finalize the matching of $k$ to $l$ and collect the reward $v_{k,l}$. Set $l$'s status to be a buyer.
      \item\label{step:buyer}\emph{If $k$ is a buyer:} Set $l$'s status to be a seller.
    \end{enumerate}
  \end{enumerate}
\end{enumerate}
\end{itemize}
\end{algorithm}

\begin{theorem}
  \label{th:factor4}
  The postponed greedy (PG) algorithm is $\nicefrac{1}{4}$-competitive for arbitrary online graphs.
\end{theorem}

\begin{proof}

  Fix a vertex $k$, and denote $p^f(s_k)$ to be the final value of its virtual seller $s_k$'s match. If $k$'s status is a seller in step (2.c.i), then we collect $p^f(s_k)$. Note that this happens with probability exactly $\nicefrac{1}{2}$ for every $k$.
  $$\text{ PG } = \E\left[\sum_{k \text{ is a seller}} p^f(s_k)\right] = \frac{1}{2} \sum_{k \in [1, T]} p^f(s_k).$$

  For a virtual buyer $b$ arriving at time $t$, let $q(b) = \max_{s \in S_t} v_{sb} - p(s)$ be the margin for $b$ in step (1.d). Note that every increase in a virtual seller's price corresponds to a virtual buyer's margin. Using the notation $S = \cup_t S_t$ and $B = \cup_t B_t$, this implies that $\sum_{s \in S} p^f(s) = \sum_{b \in B} q(b)$. 
  
The dual of the offline matching problem linear programs can be written as:
  
\begin{equation}
\begin{array}{ll@{}lll}
\text{minimize}  & \displaystyle\sum\limits_{k \in [1,T]} & \lambda_k  & &\\
\text{subject to}& & v_{kl} &\leq  \lambda_k + \lambda_l  &\forall (k,l) \text{ s.t. } |k - l| \leq d\\
& & \lambda_k &\geq 0.
\end{array}
\tag{Offline Dual}
\label{eq:offline:dual}
\end{equation}

  Let $i$ and $j > i$ be two vertices with $j - i \leq d$. When $j$ arrives, we have $q(b_j) \geq v_{ij} - p(s_i)$. Together with the fact that $p(s)$ increases over time, this implies that $\{p^f(s_k) + q(b_k)\}_{k \in [1, T]}$ is a feasible solution to \eqref{eq:offline:dual}.

  We can conclude that $\text{OFF} \leq \sum_k p^f(s_k) + q(b_k) = 2 \sum_k p^f(s_k) = 4 \text{PG}$.

\end{proof}
\subsection{Alternative algorithm for Greedy: Dynamic Deferred Acceptance}

Observe that the greedy algorithm discards a buyer that becomes unmatched and therefore does not attempt to rematch it.
We  introduce the Dynamic Deferred Acceptance (DDA) algorithm, which takes as input a constrained bipartite graph and returns a matching (formally presented below).
The main idea is to maintain a tentative  maximum-weight matching $m$ at all times during the run of the algorithm. This tentative matching is updated according to an auction mechanism:  every seller $s$ is associated with a \emph{price} $p_s$, which is initiated at zero upon arrival. Every  buyer $b$ that that already arrived and yet to become critical is associated with a \emph{profit margin} $q_b$ which corresponds to the value of matching to their most preferred seller minus the price associated with that seller.
Every time a new buyer arrives, she bids on her most preferred seller at the current set of prices. This triggers a bidding process that terminates when no unmatched buyer can profitably bid on a seller.

A tentative match between a buyer and a seller is realized (and the buyer and seller leave)  only once the seller becomes critical, i.e., she has been present for $d$ time periods and is about to become critical. At that time,  the seller and the buyer are considered matched and depart. This ensures that sellers never get matched before they become critical. A buyer is discarded only if she is unmatched and becomes  critical.

At any point $t$ throughout the algorithm, we maintain a set of sellers $S_t$, a set of buyers $B_t$, as well as a matching $m$, a price $p_s$ for every seller $s \in S_t$, and a marginal profit $q_b$ for every buyer $b \in B_t$.

\begin{algorithm}
\caption{Dynamic Deferred Acceptance}\label{alg:DDA}
  \begin{itemize}
  \item Input: an online graph with deadline $d$, $G_{d,\sigma}$.
  \item Process each event in the following way:
  \begin{enumerate}
    \item\label{step:arrival:seller}\emph{Arrival of a seller s:} Initialize $p_s \leftarrow0$ and $m(s) \leftarrow \emptyset$.
    \item\label{step:arrival:buyer}\emph{Arrival of a buyer $b$:} Start the following {\em ascending auction.} \\
    Repeat
      \begin{enumerate}
         \item Let $q_b \leftarrow \max_{s' \in S_t} v_{s', b} - p_{s'}$ and $s \leftarrow \argmax_{s' \in S_t} v_{s', b} - p_{s'}$.
         \item If $q_b > 0$ then
         \begin{enumerate}
         \item $p_s \leftarrow p_s+\epsilon$.
         \item $m(s)\leftarrow b$ (tentatively match $s$ to $b$)
         \item Set $b$ to $\emptyset$ if $s$ was not matched before. Otherwise, let $b$ be the previous match of $s$.
    	\end{enumerate}
    \end{enumerate}
    Until $q_b \leq 0$ or $b = \emptyset$.
    \item\label{step:departure:seller} \emph{Departure of a seller s:} If seller $s$ becomes critical and  $m(s) \neq \emptyset$, finalize the matching of $s$ and $m(s)$ and collect the reward of $v_{s, m(s)}$.
  \end{enumerate}
\end{itemize}
\end{algorithm}

The ascending auction phase  in our algorithm is similar to the auction algorithm by \cite{bertsekas1988auction}. Prices (for overdemanded sellers) in this auction increase by  $\epsilon$ to ensure termination, and optimality is proven  through $\epsilon$-complementary slackness conditions. For the simplicity of exposition we presented the auction algorithm but for the analysis, we consider the limit $\epsilon \rightarrow 0$ and assume the auction phase terminates with the maximum weight matching. Another way to update the matching is through the  Hungarian algorithm \cite{kuhn1955hungarian}, where  prices are increased simultaneously along an alternating path that only uses edges for which the dual constraint is tight.

The auction phase is always initiated at the existing  prices and profit margins. This, together with the fact that the graph is bipartite, ensures that prices never decrease and and marginal profits never increase throughout the algorithm. Furthermore, the prices and marginal profits of the sellers and buyers that are present in the ``market" form an optimum dual for the matching linear program (see Appendix \ref{app:missing_proofs} for more details).

In Appendix \ref{app:missing_proofs}, we show that DDA is $1/2$-competitive on constrained bipartite graphs. We note that in the case of arbitrary graphs, we can adapt the methodology of Postponed Greedy to DDA to recover a factor $1/4$.

Although the  DDA provides the same theoretical guarantees as greedy, we present it here since it may lead to better results in practice. Loosely speaking it rationalizes a reoptimization-like algorithm by keeping a tentative maximum weighted matching.

\subsection{Lower bounds}
\begin{claim}
  \label{cl:ex:bipartite_constrained}
  When the input is a constrained bipartite  graph:
  \begin{itemize}
    \item[-] No deterministic algorithm can obtain a competitive ratio above $\frac{\sqrt{5} - 1}{2} \approx 0.618$.
    \item[-] No randomized algorithm can obtain a competitive ratio above $\frac{4}{5}$.
  \end{itemize}
\end{claim}

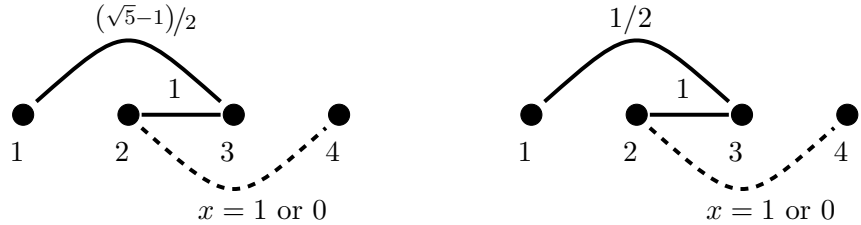
\begin{figure}[!ht]
  \centering
\begin{tikzpicture}[pre/.style={<-,shorten <=1.5pt,>=stealth,thick}, post/.style={->,shorten >=1pt,>=stealth,thick}, scale=0.7]
\tikzstyle{every node}=[draw,shape=rectangle,minimum size=5mm, inner sep=0];
\tikzstyle{edge} = [draw,thick,-]
\tikzstyle{every node}=[shape=circle,minimum size=8mm, inner sep=0];

\draw [fill](-12,0) circle [radius=0.2];
\node [right] at (-12.7,-0.7) {$1$};
\draw [fill](-10,0) circle [radius=0.2];
\node [right] at (-10.7,-0.7) {$2$};
\draw [fill](-8,0) circle [radius=0.2];
\node [right] at (-8.7,-0.7) {$3$};
\draw [fill](-6,0) circle [radius=0.2];
\node [right] at (-6.7,-0.7) {$4$};

\draw[line width=1.5pt] [-] (-10+0.25, 0) .. controls(-9,0) ..(-8-0.25, 0);
\node [right] at (-9.7,0.5) {$1$};
\draw[line width=1.5pt] [-] (-12+0.25, 0.25) .. controls(-10,1.8) ..(-8-0.25, 0.25);
\node [right] at (-10.7,1.8) {$\nicefrac{\left(\sqrt{5} - 1\right)}{2}$};
\draw[dashed, line width=1.5pt] [-] (-10+0.25, -0.25) .. controls(-8,-1.8) .. (-6-0.25, -0.25);
\node [right] at (-8.7,-1.8) {$x = 1$ or $0$};

\end{tikzpicture}
\hspace{1.5cm}
\begin{tikzpicture}[pre/.style={<-,shorten <=1.5pt,>=stealth,thick}, post/.style={->,shorten >=1pt,>=stealth,thick}, scale=0.7]
\tikzstyle{every node}=[draw,shape=rectangle,minimum size=5mm, inner sep=0];
\tikzstyle{edge} = [draw,thick,-]
\tikzstyle{every node}=[shape=circle,minimum size=8mm, inner sep=0];

\draw [fill](-12,0) circle [radius=0.2];
\node [right] at (-12.7,-0.7) {$1$};
\draw [fill](-10,0) circle [radius=0.2];
\node [right] at (-10.7,-0.7) {$2$};
\draw [fill](-8,0) circle [radius=0.2];
\node [right] at (-8.7,-0.7) {$3$};
\draw [fill](-6,0) circle [radius=0.2];
\node [right] at (-6.7,-0.7) {$4$};

\draw[line width=1.5pt] [-] (-10+0.25, 0) .. controls(-9,0) ..(-8-0.25, 0);
\node [right] at (-9.7,0.5) {$1$};
\draw[line width=1.5pt] [-] (-12+0.25, 0.25) .. controls(-10,1.8) ..(-8-0.25, 0.25);
\node [right] at (-10.7,1.8) {$1/2$};
\draw[dashed, line width=1.5pt] [-] (-10+0.25, -0.25) .. controls(-8,-1.8) .. (-6-0.25, -0.25);
\node [right] at (-8.7,-1.8) {$x = 1$ or $0$};
\end{tikzpicture}

\caption{Bipartite graph where $S = \{1, 2\}$ and $B = \{3, 4\}$, with $d = 2$: vertex $1$ becomes critical before $4$ arrives. The adversary is allowed to choose edge $(2,4)$ to be either  $1$ or $0$. Left: instance for the deterministic case. Right: instance for the randomized case.}
\label{fig:ex:constrained}
\end{figure}

\begin{proof}
  {\bf Deterministic algorithm:} Consider the example on the left of Figure \ref{fig:ex:constrained}. When seller $1$ becomes critical, the algorithm either matches her with buyer $3$, or lets $1$ departs unmatched. The adversary then chooses $x$ accordingly. Thus the competitive ratio cannot exceed:
  $$ \max \left(\min_{x \in \{0, 1\}} \frac{\frac{\sqrt{5} - 1}{2} + x}{\max(\frac{\sqrt{5} - 1}{2} + x, 1)}, \min_{x \in \{0, 1\}} \frac{1}{\max(\frac{\sqrt{5} - 1}{2} + x, 1)} \right) = \frac{\sqrt{5} - 1}{2}.$$

  {\bf Randomized algorithm:} Consider the example on the right of Figure \ref{fig:ex:constrained}. Similarly to the deterministic case, when seller $1$ becomes critical, the algorithm decides to match her with $3$ with probability $p$. The adversary then chooses $x$ accordingly. Thus the competitive ratio cannot exceed:
  $$ \max_{p \in [0,1]} \min_{x \in \{0, 1\}} \frac{p(1/2 + x) + (1-p)}{\max(1/2 + x, 1)} = 4/5.$$
\end{proof}

Next we show that our analysis for  PG is tight.
\begin{claim}
  \label{cl:tightness}
  There exists a \emph{constrained} bipartite graph for which PG is $\nicefrac{1}{(4 - 2 \epsilon)}$ -competitive.
\end{claim}

\begin{figure}[!ht]
  \centering

\def\radius{2.6}
\def \Pointsize {1.4pt}
\begin{tikzpicture}[pre/.style={<-,shorten <=1.5pt,>=stealth,thick}, post/.style={->,shorten >=1pt,>=stealth,thick}, scale=0.7]
\tikzstyle{every node}=[draw,shape=rectangle,minimum size=5mm, inner sep=0];
\tikzstyle{edge} = [draw,thick,-]
\tikzstyle{every node}=[shape=circle,minimum size=8mm, inner sep=0];

\draw [fill](-12,0) circle [radius=0.2];
\node [right] at (-12.7,-0.7) {$1$};
\draw [fill](-10,0) circle [radius=0.2];
\node [right] at (-10.7,-0.7) {$2$};
\draw [fill](-8,0) circle [radius=0.2];
\node [right] at (-8.7,-0.7) {$3$};
\draw [fill](-6,0) circle [radius=0.2];
\node [right] at (-6.7,-0.7) {$4$};

\draw[line width=1.5pt] [-] (-10+0.25, 0) .. controls(-9,0) ..(-8-0.25, 0);
\node [right] at (-9.7,0.5) {$1$};
\draw[line width=1.5pt] [-] (-12+0.25, 0.25) .. controls(-10,1.8) ..(-8-0.25, 0.25);
\node [right] at (-10.7,1.8) {$1 - \epsilon$};
\draw[dashed, line width=1.5pt] [-] (-10+0.25, -0.25) .. controls(-8,-1.8) .. (-6-0.25, -0.25);
\node [right] at (-8.7,-1.8) {$1$};
\end{tikzpicture}

\caption{Bipartite graph where $S = \{1, 2\}$ and $B = \{3, 4\}$, with $d = 2$: vertex $1$ becomes critical before $4$ arrives. Dotted edges represent edges that are not know to the algorithm initially.}
\label{fig:ex:tightness}

\end{figure}
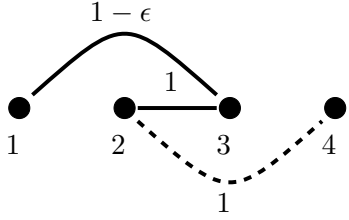

\begin{proof}
Consider the input graph in Figure \ref{fig:ex:tightness}.
Vertex $2$ will be temporarily matched with $3$, and vertex $1$ will depart unmatched. When $2$ becomes critical, with probability $\nicefrac{1}{2}$, she will be determined to be a \emph{buyer} and will depart unmatched. Therefore the PG collects in expectation $\nicefrac{1}{2}$ while the offline algorithm collects $2 - \epsilon$.

\end{proof}

\section{Random arrival order}
In some cases, the vertices can be assumed to come from a distribution that is unknown to the online algorithm. One way to model this is to assume that the adversary chooses the underlying graph, but that the  vertices arrive in random order. 

\subsection{The batching algorithm}

The {\it batching} algorithm computes a maximum-weight matching every $d+1$ time steps.  Every vertex in the matching is then matched, and all other vertices in the batch are discarded.  




\begin{theorem}
  \label{th:competitive_ratio}
	Batching is $\left(0.279 + O(1/n)\right)$-competitive.
\end{theorem}

The proof of Theorem \ref{th:competitive_ratio} works in three steps. In a first step, we reduce the analysis of the competitive ratio of \emph{Batching} to a graph covering problem. More precisely, we show that it is enough to cover $C_n^d$, the cycle with $n$ vertices to the power $d$, with ensembles of cliques. Second, we show how a cover for small $n$ can be extended to any $n$ at the cost of a small rounding error. Finally, we establish  a reduction that allows us to consider only a finite set of values for $d$. We conclude with a computer-aided argument for graphs in the finite family.

\subsubsection*{Reducing to a graph theoretic problem}

There is no harm in assuming that the underlying graph $G$ is a complete. 
Recall that $S_n$ is the set of all permutations over integers $1, ..., n$.
For any deadline $d$ and any arrival sequence $\sigma \in S_n$, we define the {\it path} graph $P_n^d(\sigma)$ with edge-weight $v_{ij} = 1$ if $|\sigma(i) - \sigma(j)| \leq d$, and $v_{ij} = 0$ otherwise.\footnote{Note that $P_n^d(\sigma)$ corresponds to the path $(\sigma(1), \sigma(2)), (\sigma(2), \sigma(3)), ..., (\sigma(n-1), \sigma(n))$ taken to the power $d$.}

Note that every batch in the algorithm has $d+1$ vertices except  the last batch which may have fewer vertices. Let $b_i(\sigma, d)$ be the batch of vertex $i$ under permutation $\sigma$ and batch size $d+1$: $b_i(\sigma, d)$ is the unique integer such that $ (d + 1) (b_i - 1)  < \sigma(i) \leq (d+1)b_i.$
We define the {\it batched} graph $B_n^{d}(\sigma)$ with edge-weight $v_{ij} = 1$ if $i$ and $j$ are in the same batch (i.e. $b_i(\sigma, d+1) = b_j(\sigma, d+1)$), and $v_{ij} = 0$ otherwise.\footnote{Note that $B_n^{d}(\sigma)$ is a collection of disjoint $(d+1)$-cliques.}

For any $n \geq d \geq 1$,  denote $C_n^d$ to be the $n$-cycle to the power $d$. 

\begin{definition}[Graph operations]
  For any two graphs $H$ and $H'$ with vertices $1, ..., n$ and respective edge weights $v_{ij}, v'_{ij}$, we define the following:
  \begin{itemize}
    \item[(i)] The linear combination $a H + b H'$ denotes the graph with edge weights $a v_{ij} + b v'_{ij}$,
    \item[(ii)] The product $H * H'$ denotes the graph with edge weights $v_{ij} * v'_{ij}$, and
    \item[(iii)] We say that $H$ is a {\it cover} of $H'$ if for all $i$, $j$, $v_{i,j} \geq  v'_{ij}$.
  \end{itemize}
\end{definition}

For any graph $H$, let $m(H)$ denote the value of a maximum-weight matching over $H$. Observe that when the arrival sequence is $\sigma$, the graph $P_n^d(\sigma) * G=G(d,\sigma)$ and therefore  the offline algorithm collects $m(P_n^d(\sigma) * G)$. Note that the  online algorithm collects $m(B_n^{d}(\sigma) * G)$.

\begin{remark}
  \label{rem:2.2}
  Observe that for any graphs $H, H', G$ and any $a, b \in \mathbb{R}$, we have:
  \begin{itemize}
    \item[-] $m(a H + b H') \leq a m(H) + b m(H')$.
    \item[-] If $H$ is a cover of $H'$, then, $m(H * G) \geq m(H' * G)$.
  \end{itemize}
\end{remark}

\begin{definition}[Periodic permutation]
  For $p < n$ such that $p$ divides $n$, we say that a permutation $\sigma \in S_n$ is $p$-periodic if 
   for all $i \in [1,n-p]$, $\sigma(i + p) \equiv \sigma(i) + p \mod n$.

  We say that a permutation $\sigma$ is periodic if there exists $p$ such that $\sigma$ is $p$-periodic.
\end{definition}

\begin{definition}[$(\alpha,d)$-cover]
  Let $F$ be an unweighted graph with $n$ vertices. We say that a set of permutations $\{\sigma_1, ..., \sigma_K\} \in S_n$ forms an $(\alpha, d)$-cover of $F$ if there exist values $\lambda_1, ..., \lambda_K \in [0,1]$ such that:
  \begin{itemize}
    \item[(i)] $\sum_{k \leq K} \lambda_k B_n^d(\sigma_k)$ is a cover of $F$.
    \item[(ii)] $\sum_{k \leq K} \lambda_k = \alpha$.
  \end{itemize}
  We say that an $(\alpha, d)$-cover is $p$-periodic if for all $k$, $\sigma_k$ is $p$-periodic.
\end{definition}



%
%

The next proposition  will allow us to abstract away from the weights that are chosen by the adversary.  For any graph $H$, we denote by $H_{ij}$  the weight $v_{ij}$ in $H$.
\begin{proposition}
  \label{prop:cover}
  If there exists an $(\alpha, d)$-cover of $C_n^d$, then batching is $1/\alpha$-competitive.
\end{proposition}

\begin{proof}
  Let {\it id} be the identity permutation over $n$ vertices. Let $\{ \sigma_1, ..., \sigma_K \}$ be an $(\alpha, d)$-cover of $C_n^d$.
  Fix an arrival sequence $\sigma \in S_n$. We first claim that  $\{ \sigma_1 \circ \sigma, ..., \sigma_K \circ \sigma \}$ is an $(\alpha, d)$-cover of $P_n^d(\sigma)$.

   For any $\sigma \in S_n$, let us denote $\beta_{i,j}(\sigma)$ and $\rho_{i,j}(\sigma)$ to be the weights of edge $(i,j)$ in $B_n^d(\sigma)$ and $P_n^d(\sigma)$ respectively. Consider $(i, j) \in P_n^d(\sigma)$: $|\sigma(i) - \sigma(j)| \leq d$:

  \begin{equation*}
    \begin{split}
      \sum_k \lambda_k \beta_{i,j}(\sigma_k \circ \sigma) &= \sum_k \lambda_k \mathbb{I}[b_i(\sigma_k \circ \sigma, d)= b_j(\sigma_k \circ \sigma, d)] \\
      & = \sum_k \lambda_k \mathbb{I}[b_{\sigma(i)}(\sigma_k, d) = b_{\sigma(j)}(\sigma_k, d)]\\
      & \geq \rho(\text{id})_{\sigma(i), \sigma(j)} = 1,
    \end{split}
  \end{equation*}
  where the last inequality is implied by the fact that  $\{ \sigma_1, ..., \sigma_K \}$ is an $(\alpha,d)$-cover of $C_n^d$ and therefore of $P_n^d(\text{id})$.
Therefore the claim holds using remark \ref{rem:2.2}.

Denote by BAT the value collected by the  batching algorithm and OFF the value collected by the offline algorithm. Observe that
  \begin{equation*}
    \begin{split}
       \text{OFF} &= \frac{1}{n!}\sum_{\sigma \in S_n} m(P_n^{d}(\sigma) * G) \\
        &\leq \frac{1}{n!} \sum_{\sigma \in S_n} \sum_k \lambda_k m(B_n^{d}(\sigma_k \circ \sigma) * G) \\
       &= \frac{1}{n!} \sum_k \lambda_k \sum_{\sigma' \in S_n} m(B_n^{d}(\sigma') * G) \\
       &= \alpha \text{BAT},
     \end{split}
   \end{equation*}

where we used the change of variable $\sigma' = \sigma_k \circ \sigma$ and the fact that the application $\mathcal{A}_k: \sigma \mapsto \sigma_k \circ \sigma$ is a bijection. 

\end{proof}


We have reduced the analysis of Batching to a graph-theoretic problem without edge weights. In what follows, we will show that we can reduce the problem further to find covers of $C_n^d$ for only small values of $n$ and $d$.

\subsubsection*{Reducing $n$: periodic covers.}
We now wish to find $(\alpha, d)$-covers for $C_n^d$ for every $n$ and $d$. In Proposition \ref{prop:cycle}, we show that it is sufficient to find periodic covers for small values of $n$.

\begin{proposition}
  \label{prop:cycle}
  Let $p$ be a multiple of $d+1$, and $n_1$ a multiple of $p$. Any $p$-periodic $(\alpha, d)$-cover of $C_{n_1}^d$ can be extended into an $(\alpha + O(\nicefrac{p}{n}), d)$-cover of $C_{n}^d$ for any $n \geq n_1$. 
\end{proposition}

\begin{proof}[Proof when $n$ is a multiple of $p$.]
  Let $\{\sigma_1, ..., \sigma_K\}$ be a $p$-periodic $(\alpha, d)$-cover of $C_{n_1}^d$. We will show that it can be extended into an $(\alpha, d)$-cover of $C_n^d$.

  Assume for now that $n$ is a multiple of $p$. Let $\sigma'_k$ be the $p$-periodic permutation over $1, ..., n$ such that for all $i \in [1, p]$, $\sigma'_k(i) = \sigma_k(i)$. Take $i',j' \in [1, n]$ such that $|i' - j' | \leq d$.
  Because $n_1>p$ is a multiple of $p$, there exist $i, j \in [1, n_1]$ such that $i \equiv i' \mod p$, $j \equiv j' \mod p$ and $| i - j| \leq d $.
  By $p$-periodicity of $\sigma_k$ and $\sigma'_k$, we know that $B_n^{d}(\sigma'_k)_{i',j'} = B_{n_1}^{d}(\sigma_k)_{i,j}$. Thus we can conclude that $\{\sigma_1', ..., \sigma_K'\}$ is an $(\alpha, d)$-cover of $C_n^d$.
\end{proof}

 In the case when $n$ is not a multiple of $p$, the proof follows similar ideas and looses an additional factor $\left(\frac{n}{n-p}\right)^2$ due to rounding of $n$ to a lower multiple of $p$. Details are provided in Appendix \ref{app:ro_proofs}.

  

\subsubsection*{Reducing $d$: cycle contraction.}
In Proposition \ref{prop:cycle}, we show that it is enough to find periodic $(\alpha, d)$-covers of $C_n^d$ for small values of $n$. Next, we provide a reduction that enables us to consider only a finite set of values for $d$.

The key idea of the reduction is that we can contract vertices of $C_{n}^d$ into $n/u$ groups of $u$ vertices. 
The resulting graph also happens to be a cycle $C_{n/u}^{(d+1)/u}$. 
 In Proposition \ref{prop:non_multiples}, we provide a way to expand an $(\alpha, u-1)$-cover on the contracted graph into an $(\alpha, d)$ cover on the original graph.

\begin{definition}[Cycle contraction]
	\label{def:contraction}
  For any $n,d$ and an integer $u$ which divides $n$, we define the $u$-contraction $f_u(C_n^d)$ to be the graph with vertices $a_k = \{uk+1, ..., u(k+1)\}$ for $k \in [0, \nicefrac{n}{u}-1]$, and edges $(a_k, a_l)$ if and only if there exist $i \in a_k$ and $j \in a_l$ with an edge $(i,j)$ in $C_n^d$.
\end{definition}

\begin{claim}
  For any $d$, if $u > 1$ divides $d+1$ and $d+1$ divides $n$, then $f_u(C_n^d) = C_{n/u}^{(d+1)/u}$.
\end{claim}

\begin{proof}
  We first prove that  $C_{n/u}^{(d+1)/u}$ covers $f_k(C_n^d)$. Fix $k,l \in [0, \nicefrac{n}{u}-1]$, and assume that $k < l$. If $|l-k| \leq (d+1)/u$, then let $i = u(k+1)$ and $j = ul + 1$. We have $|j - i| = u(l - k - 1) + 1 \leq d$, thus $(i,j) \in C_n^d$ and $(k, l) \in f_u(C_n^d)$.

  Conversely, we now prove that $f_u(C_n^d)$ covers $C_{n/u}^{(d+1)/u}$. If there exist $i \in a_k$ and $j \in a_l$ such that $|j - i| \leq d$, then $u(l - k) \leq ul + 1 - u(k+1) \leq d+1$ which implies that $(k,l) \in  C_{n/u}^{(d+1)/u}$.
\end{proof}

\begin{figure}[H]
  \centering
 \includegraphics[height=5cm]{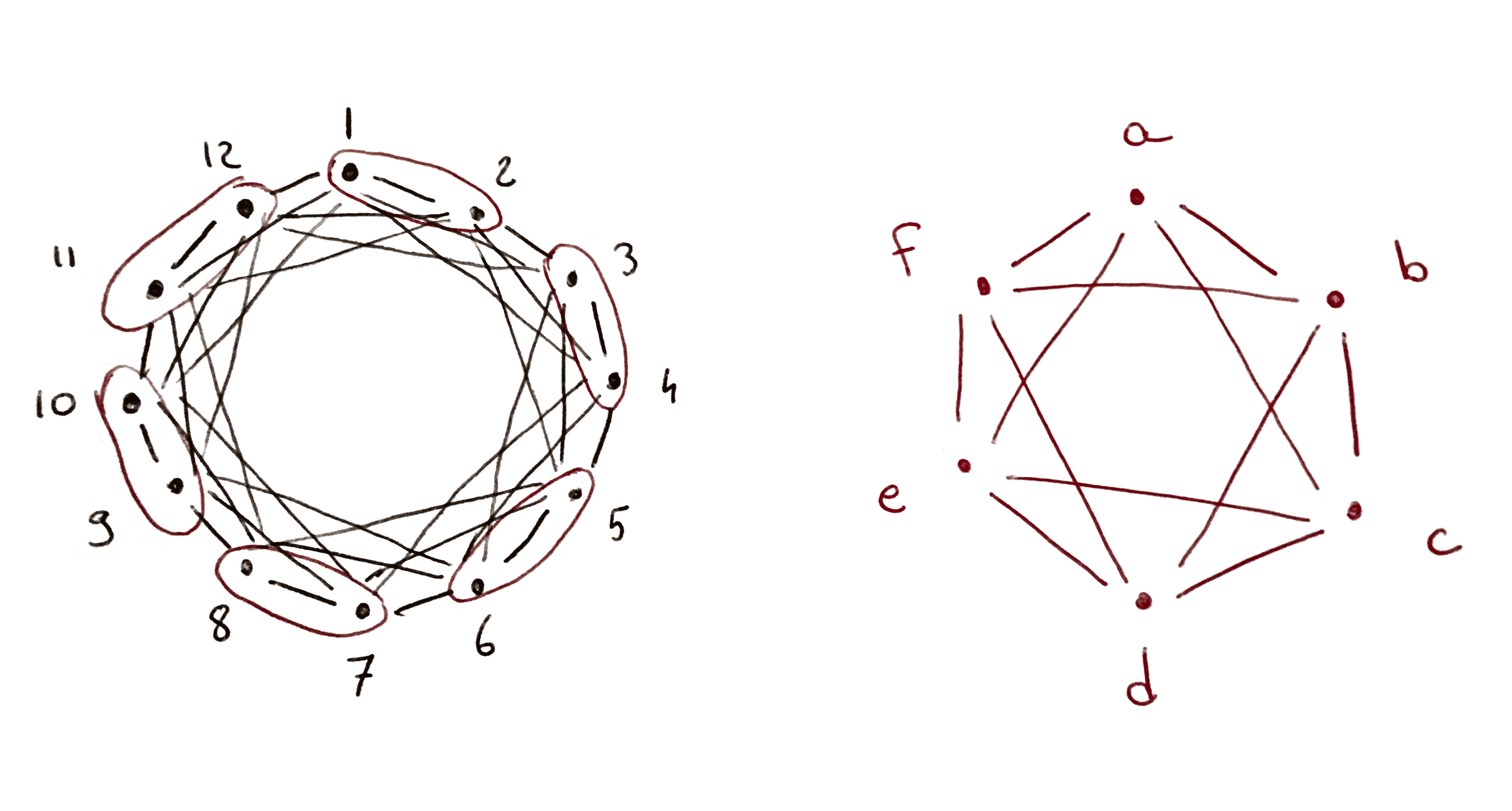}
 \caption{Left: $C_{12}^3$, with contraction for $u = 2$.  Right: Contracted graph $f(C_{12}^3) = C_6^2$ with vertices $a = \{1, 2\}$, $b = \{3, 4\}$, ... $f = \{11, 12\}$.}
 \label{fig:contraction}
 \hspace{0.2cm}
\end{figure}

\begin{proposition}
  \label{prop:non_multiples}
  Fix $d \geq 1$. For $d+1 > k\geq 1$, suppose that there is a periodic $(\alpha, k-1)$-cover of $C_{rk}^k$.
  \begin{itemize}
  \item[(i)] For any integer $r$, if $k$ divides $d+1$ then there exists a periodic $\left(\alpha, d\right)$-cover of $C_{r(d+1)}^d$.
  \item[(ii)] In general, if $v$ is the remainder of the euclidian division of $d+1$ by $k$, then there exists a periodic $\left(\alpha (1 + \nicefrac{v}{d+1-v})^2, d\right)$-cover of $C_{r(d+1)}^d$.
  \end{itemize}
  
\end{proposition}



\begin{proof}[Proof of (i)]
  Suppose that $d+1 = ku$ and suppose that there exists $p$ multiple of $d+1$ such that we have a $p$-periodic $(\alpha, k-1)$-cover $\{\sigma_1, ..., \sigma_K\}$ of $f_{u}\left(C_{r(d+1)}^d\right) = C_{rk}^k$.
  For any permutation $\sigma \in S_{rk}$ we can construct a permutation $\sigma' \in S_{r(d+1)}$ in the following way: if $i \in a_t$ then $\sigma'(i) = \frac{n}{k} \sigma(t) + i$.
  Because $B_{rk}^k(\sigma_i)$ is a cover of $B_{r(d+1)}^d$, we can conclude that $\sigma'_1, ..., \sigma'_K$ is an $(\alpha, d)$-cover of $C_{r(d+1)}^d$.

\end{proof}

The proof for case (ii) follows a similar idea, with an additional randomization that chooses a subset $d+1-v$ vertices that we can group in every group of $d+1$. The details are in Appendix \ref{app:ro_proofs}.

\subsubsection*{Final step: Computer-aided proof of factor $2.79$}

We will now apply Proposition \ref{prop:cycle} with $p = 2(d + 1)$ and $n_1 = 4(d+1)$. Let $\Omega_d$ be the set of $2(d+1)$-periodic permutations of $1, ..., 4(d+1)$.
We can find covers for $C_{4(d + 1)}^d$ using the following linear program:

\begin{equation*}
  \label{eq:LP}
\begin{array}{ll@{}ll}
\text{minimize}  & \displaystyle\sum\limits_{\sigma \in \Omega_{d}} &\lambda_\sigma  &\\
\text{subject to}& \displaystyle\sum\limits_{\sigma \in \Omega_{d}}   &\lambda_{\sigma} \mathbb{I}[b_i(\sigma, d) = b_j(\sigma, d)] \geq 1,  &\forall (i,j) \in C_{4(d + 1)}^d\\
                 &                                                &\lambda_{\sigma} \in \mathbb{R}^+, &\sigma\in \Omega_{d}
\end{array}
\tag{LP$_d$}
\end{equation*}
\begin{proposition}
Let $\alpha_d$ be the solution to LP$_d$. Let $\alpha = \sup_{d \geq 1} \alpha_d$. Batching is $\left(1 / \alpha + O(1/n)\right)$-competitive.
\end{proposition}
\begin{proof}
Follows from Propositions \ref{prop:cover} and \ref{prop:cycle}.
\end{proof}
The Linear program \eqref{eq:LP} has $O(d!)$ variables, and solving it may not be computationally possible when $d$ is large. Using Proposition \ref{prop:non_multiples}, we now provide a way to find upper bounds on $\alpha_d$ by solving a different LP on a smaller graph.

Recall that $\Omega_{k-1}$ is the set of $2k$-periodic permutations of $1, ..., 4k$. We define the problem of finding an $(\alpha, k-1)$-cover of the cycle $C_{4k}^k$. 
\begin{equation*}
  \label{eq:LP'}
\begin{array}{ll@{}ll}
\text{minimize}  & \displaystyle\sum\limits_{\sigma \in \Omega_{k-1}} &\lambda_\sigma  &\\
\text{subject to}& \displaystyle\sum\limits_{\sigma \in \Omega_{k-1}}   &\lambda_{\sigma} \mathbb{I}[b_i(\sigma, k-1) = b_j(\sigma, k-1)] \geq 1,  &\forall (i,j) \in C_{4k}^k\\
                 &                                                &\lambda_{\sigma} \in \mathbb{R}^+, &\sigma\in\Omega_{k-1}
\end{array}
\tag{LP'$_k$}
\end{equation*}

We denote by $\alpha'_k$ the solution to \eqref{eq:LP'}.
 Solving \eqref{eq:LP'} numerically for $k = 4$ yields $\alpha'_4 \leq 3.17$. For all $d \geq 52$ Proposition \ref{prop:non_multiples} therefore implies that, $\alpha_d \leq 3.17 * \left(\frac{51}{49}\right)^2 = 3.58$. \footnote{We note that our methodology can be extended to obtain a better factor. For instance, being able to solve \eqref{eq:LP} for values higher than $50$ would lead to a competitive ratio closer to $\frac{1}{3}$.}
 For $d \leq 50$, we either solve \eqref{eq:LP} directly, or use Proposition \ref{prop:non_multiples} to show that $\alpha_d \leq 3.58$ (see Appendix \ref{app:numerical}).
  Observing that $2.79 \leq \frac{1}{3.58}$, this concludes the proof for Theorem \ref{th:competitive_ratio}.
  \qedsymbol

\subsection{Lower bound in random order.}

\begin{proposition}
  No algorithm is more than $ \frac{1}{2}$-competitive even under the random arrival order.
\end{proposition}

\begin{proof}
Consider a graph with three vertices $\{1,2,3\}$ and $d = 1$, i.e. vertices can only be matched to the ones arriving just before or after them.
After the first two arrivals, the online algorithm $\mathcal{A}$ needs to decide whether to match them or let the first arrival leave. Furthermore, it has no information on how $v_{\sigma(1), \sigma(2)}$ compares to the other edge weights. Therefore the decision of whether to match has to be a coin toss.
Regardless of whether the algorithm matches the first two or the last two arrivals, its expected reward is $\frac{v_{1,2} + v_{2,3} + v_{3,1}}{3}$. $\text{OFF}$ however has an expected reward of $\frac{\max(v_{1,2}, v_{2,3}) + \max(v_{2,3}, v_{3,1}) + \max(v_{3,1}, v_{1,2})}{3}$.
Taking $v_{1,2} = v_{2,3} \rightarrow 0$, we get $\mathcal{A} = 1/3$ while $\text{OFF} = 2/3$ which concludes the proof.
\end{proof}


\section{Extensions}

\subsection{Stochastic departures in the adversarial order setting}

We relax the assumption that all vertices depart after exactly $d$ time steps.

We therefore focus on the stochastic case, in which  the departure time $d_i$ of vertex $i$ is sampled independently from a distribution $\mathcal{D}$.
We assume that the realizations $d_i$ are only known at the time $i$ becomes critical.

\begin{proposition}
  \label{prop:single_bid}
  Suppose that there exists $\alpha \in (0,1)$ such that $\mathcal{D}$ satisfies the property that for all $i < j$,
  $$\P[i + d_i \leq j + d_j | i + d_i \geq j] \geq \alpha.$$
  Then PG is $\nicefrac{\alpha}{4}$-competitive.
\end{proposition}

\begin{proof}
When a vertex $k$ becomes critical in the original graph, we match it if its status is determined to be \emph{seller}. In that case, we need to ensure that its tentative match $b_l$ is still present. With probability at least $\alpha$, vertex $l$ is still present, and we collect $p(s_k)$. The rest of the proof follows similarly to that of Theorem \ref{th:factor4}
\end{proof}

\begin{corollary}
  PG is $\nicefrac{1}{8}$-competitive when $\mathcal{D}$ is memoryless.
\end{corollary}

\subsection{Look-ahead under random arrival order}


We assume now that the online algorithm knows vertices that will arrive in $l$ time steps (and their adjacent edges). We can update the Batching Algorithm in the following way: every $d+l+1$ time steps, compute a maximum-weight matching on both the current vertices and the next $l$ arrivals. Match vertices as they become critical according to the matching, and discard unmatched vertices. 
Note that this is the same as running Batching when the deadline is $d+l$.

\begin{proposition}
There exists an $(\frac{d+l+1}{l+1}, d + l)$-cover of $C_n^d$.
\end{proposition}

\begin{proof}
    For $k \in [0, d+l]$, let $\sigma_k(i) = i + k\mod n$. Let $i,j$ be such that $|i - j| \leq d$, then $b_i(\sigma_k, d) = b_j(\sigma_k, d)$ for at least $l+1$ different values of $k$.
    We can conclude that $\sigma_0, ..., \sigma_{d + l}$ is a $(\frac{d + l + 1}{l+1}, d+l)$-cover by taking $\lambda_0 = ... = \lambda_{d + l} = \frac{1}{l + 1}$.
\end{proof}

\begin{corollary}
  Batching with $l$-lookahead is $\frac{l+1}{d+l+1}$-competitive when $n$ is large.
\end{corollary}


\section{Conclusion}

This  paper introduces a model for dynamic matching in which all agents arrive and depart after some deadline. Match values are heterogeneous and the underlying graph is  non-bipartite. We study online algorithms for two settings, where vertices  arrive  in an adversarial or random order.

In the adversarial arrival case, we introduce two new $1/4$-competitive algorithms when departures are deterministic and known in advance. We also provide a $1/8$-competitive algorithm when departures are stochastic, i.i.d, memoryless, and known at the time a vertex becomes critical. Finally we show that no online algorithm is more than $1/2$-competitive.

In the random arrival case, we show that a batching algorithm is $0.279$-competitive. We also show that with knowledge of future arrivals, its performance guarantee increases towards $1$.

Importantly,  our model imposes restrictions on the departure process and requires the  algorithm to know when vertices become critical. Other than closing the gaps between the upper bound  $\nicefrac{1}{2}$ and the achievable competitive ratios, we point out a just a few interesting directions for  future research. 
Our model imposes that matches retain the same value regardless of when they are conducted. An interesting direction is to account for agents' waiting times.  
A different  intersting  objective is to achieve both a high total value and a large  fraction of matched agents. Finally, it is interesting to consider the stochastic setting with prior infromation over weights and future arrivals.


%
%

\bibliography{bibliography}
\bibliographystyle{apalike}

\appendix

\section{Missing proofs for Dynamic Deferred Acceptance}
\label{app:missing_proofs}
\begin{lemma}
\label{lem:mon}
Consider the DDA algorithm on a constrained bipartite graph.
\begin{enumerate}
\item Throughout the algorithm,  prices corresponding the sellers never decrease and the profit margins of buyers never increase.
\item At the end of every ascending auction,  prices of the sellers and the marginal profits of the buyers form an optimal solution to the dual  of the matching linear program associated with  buyers and sellers present at that particular time.
\end{enumerate}
\end{lemma}

Maintaining a  maximum-weight matching along with optimum dual variables does not guarantee an efficient matching for the whole graph. The dual values are not always feasible for the offline problem. Indeed, the profit margin of some  buyer $b$ may  decrease after some seller departs the market. This is because $b$ may face increasing competition from new buyers, while the bidding process excludes sellers that have already departed the market (whether matched or not).

\begin{proposition}
  \label{prop:dda_factor_2}
  DDA is $\nicefrac{1}{2}$-competitive for constrained bipartite  graphs.
\end{proposition}





\begin{proof} The proof follows the primal-dual framework.

First, we observe that by complementary slackness, any seller $s$ (buyer $b$) that departs unmatched has a final price $p_s^f = 0$ (final profit margin $q_b^f = 0$). When a seller $s$ is critical and matches to $b$, we have  $v_{s,b} = p_s^f + q_b^f$. Therefore, \emph{DDA} collects a reward of $\mathcal{A} = \sum_{s \in S} p_s^f + \sum_{b \in B} q_b^f$.

Second, let us consider a buyer $b$ and a seller $s \in [b - d, b)$ who has arrived before $b$ but not more than $d$ steps before. Because sellers do not finalize their matching before they are critical, we know that $s \in S_b$. An ascending auction may be triggered at the time of $b$'s arrival, after which we have: $v_{s,b} \leq p_s(b) + q_b(b) \leq p_s^f + q_{b}^i$, where the second inequality follows from the definition that $q_b(b) = q_b^i$ and from the monotonicity of sellers' prices (Lemma \ref{lem:mon}). Thus, $(p^f, q^i)$ is a feasible solution to the offline dual problem.

Finally, we observe that upon the arrival of a new buyer, the ascending auction does not change the sum of prices and margins for vertices who were already present:

\begin{claim}
  \label{cl:monotonicity}
  Let $b$ be a new buyer in the market, and let $p, q$ be the prices and margins before $b$ arrived, and let $S_t$ and $B_t$ be the set of sellers and buyers present before $b$ arrived. Let $p'$, $q'$ be the prices and margins at the end of the ascending auction phase (Step 2(a) in Algorithm 1). Then:
\begin{equation}
\sum_{s \in S_t} p_s + \sum_{b \in B_t} q_b = \sum_{s \in S_t} p'_s + \sum_{b \in B_t} q'_b.
\label{eq:conservation}
\end{equation}
\end{claim}
By applying this equality iteratively after each arrival, we can relate the initial margins $q^i$ to the final margins $q^f$ and prices $p^f$:
\begin{claim}
  \label{cl:conservation}
  $\sum_{s \in S} p_s^f + \sum_{b \in B} q_b^f = \sum_{b \in B} q_b^i$.
\end{claim}

This completes the proof of Proposition \ref{prop:factor2} given that the offline algorithm achieves at most:
$$\mathcal{O} \leq \sum_{s \in S} p_s^f + \sum_{b \in B} q_b^i \leq 2 \mathcal{A}.$$

\end{proof}

It remains to prove Claims \ref{cl:monotonicity} and \ref{cl:conservation}.

\begin{proof}[Proof of Claim \ref{cl:monotonicity}]
  The proof of termination in \cite{bertsekas1988auction} relies on the introduction of a minimum bid $\epsilon$ in step $6$ of the auction algorithm to ensure that the algorithm does not get stuck in a cycle of bids of $0$. In the limit where $\epsilon \rightarrow 0$, the algorithm ressembles the \emph{hungarian algorithm} \cite{kuhn1955hungarian}. The idea is to search for an augmenting path along the edges for which the dual constraint is tight. If such a path is found, the matching is augmented, otherwise we perform simultaneous bid increases in way that ensures that prices $p$ and margins $q$ are still dual feasible.

  We assume that we are given at time $t$ an optimal matching $m$ and optimal duals $(p, q)$ corresponding to the graph with vertices $S_t, B_t$. We assume that we added a new vertex $b^*$ to $B'_{t} = B_t \cup \{b^*\}$, and that we initialized $q_{b^*} = \max_{s \in S_t} v_{s, b^*} - p_s$ 

  Initialize $m' = m$, $p' = p$, $q' = q$. Note that primal and dual feasibility are satisfied. Therefore, $(m', p', q')$ is optimal iff the following three complementary slackness condition are satisfied:
\begin{equation}
  \label{eq:CS:dual}
  \forall s \in S_t, v_{s,m(s)} = p'_s + q'_m(s).
  \tag{CS1}
\end{equation}
\begin{equation}
  \label{eq:CS:primal}
  \forall s \in S'_t, m(s) = \emptyset \implies p_s = 0.
  \tag{CS2}
\end{equation}
\begin{equation}
  \label{eq:CS:primal2}
  \forall b \in B'_t, m(b) = \emptyset \implies q_b = 0.
  \tag{CS3}
\end{equation}
  Note that \eqref{eq:CS:dual} and \eqref{eq:CS:primal} are already satisfied. If $q'_{b^*} = 0$ then \eqref{eq:CS:primal2} is also satisfied and we have an optimal solution.

  Suppose now that $q'_{b^*} > 0$. We will update $(m', p', q')$ in a way that maintains primal and dual feasibility, as well as \eqref{eq:CS:dual} and \eqref{eq:CS:primal}.

  Our objective is to find an augmenting path in the graph. First we will start by trying to find an alternating path that starts on $b$ and only uses edges for which the dual constraint is tight: $\mathcal{E} = \{(s,b) | s \in S'_t, b \in B'_t,  v_{s,b} = p'_s + q'_b\}$. Observe that by \eqref{eq:CS:dual} all the matched edges in $m$ are in $\mathcal{E}$. We will now successively color vertices as follows:
  \begin{itemize}
    \item[0.] Start by coloring $b^*$ in blue.
    \item[1.] For any blue buyer $b$, for any seller $s$ such that $(s, b) \in \mathcal{E}$ and $s \neq m(b)$, we color $s$ in red.
    \item[2.] For any red seller $s$, let $b = m(s)$, then color $b$ in blue.
  \end{itemize}

  Observe that there is an alternating path between $b^*$ and any red seller.
  If at one point we color an unmatched seller $s^*$ in red, this means that we have found an augmenting path from $b^*$ to $s^*$ that only utilizes edges in $\mathcal{E}$. In that case, we change $m'$ according to the augmenting path. Because of the way we chose edges in $\mathcal{E}$, \eqref{eq:CS:dual} is still satisfied. \eqref{eq:CS:primal} and \eqref{eq:CS:primal2} are now also satisfied, which means we have an optimal solution $(m', p',q')$.

  We terminate when we are unable to color vertices any further. In that case, let us define $\delta_1 = \min_{b \text{ blue}} q_b$. If $\delta_1 = 0$, then there exists $b \in B_{t}'$ with $q_b = 0$ and an alternating path form $b^*$ to $b$. We update $m'$ according to that path, and verify that all CS conditions are now satisfied.

  Suppose that $\delta_1 > 0$. Define
  \begin{equation}
    \delta_2 = \min_{b \text{ blue, } s \text{ not red}} \{ p_s + q_b - v_{s,b} \}.
    \label{eq:delta2}
  \end{equation}

  The fact that we cannot color any more vertices implies that $\delta_2 > 0$.
  Let $\delta = \min(\delta_1, \delta_2) > 0$. For every red seller $s$, we update $p'_s \leftarrow p'_s + \delta$. For every blue buyer $b$, we update $q'_b \leftarrow q'_b - \delta$. Observe that dual feasibility is still verified, as well as \eqref{eq:CS:dual}.

  If $\delta = \delta_2$, taking $(s,b)$ the argmin in \eqref{eq:delta2}, we now have such that $p'_s + q'_b - v_{s,b} = 0$ which means we can add $(s,b)$ to $\mathcal{E}$ and color $s$ in red.
  We will eventually have $\delta = \delta_1$, and this leads to $q_b = 0$ and we can terminate.
  This proves both the termination and correctness. Furthermore, monotonicity of the dual variables is also straightforward. Let us now prove the conservation property:
  \begin{equation}
    \label{eq:proof:conservation}
    \sum_{s \in S_t} p_s + \sum_{b \in B_t} q_b = \sum_{s \in S_t} p'_s + \sum_{b \in B} q'_b.
  \end{equation}
  Note that when we update the dual variables, then every seller we colored in red was matched in $S'$ and we colored that match in blue. Therefore, apart from the initial vertex $i$, there are the same number of red and blue vertices.
\end{proof}

\begin{proof}[Proof of Claim \ref{cl:conservation}]
The idea of the proof is to iteratively apply the result of Claim \ref{cl:monotonicity} after any new arrival.
Let $\widetilde{S}_t$ (resp. $\widetilde{B}_t$) be the set of sellers (buyers) who have departed, or already been matched before time $t$. We show by induction over $t \leq T$ that:

\begin{equation}
  \sum_{s \in \widetilde{S}_t} p_s^f + \sum_{b \in \widetilde{B}_t} q_b^f + \sum_{s \in S_t} p_s(t) +  \sum_{b \in B_t} q_b(t) = \sum_{b \in \widetilde{B}_t} q_b^i +  \sum_{b \in  B_t} q_b^i.
  \label{eq:proof:balance}
\end{equation}
This is obvious for $t = 1$. Suppose that it is true for $t \in [1, T-1]$. Note that departures do not affect \eqref{eq:proof:balance}. If the agent arrivint at $t + 1$ is a seller, then for all other sellers $s$, $p_s(t+1) = p_s(t)$ and for all buyers $b$, $q_b(t+1) = q_b(t)$, thus \eqref{eq:proof:balance}, is clearly still satisfied. Suppose that vertex $t + 1$ is a buyer. Using equation \eqref{eq:conservation}, we have:

$$\sum_{s \in S_{t+1}} p_s(t+1) +  \sum_{b \in B_{t+1}} q_b(t+1) = q_{t+1}(t+1) + \sum_{b \in B_t} q_b(t) + \sum_{s \in S_t} p_s(t) = \sum_{b \in  B_{t+1}} q_b^i.$$

Note that at time $T + d$, every vertex has departed. Thus, $\widetilde{S}_{T+d} = S$, $\widetilde{B}_{T+d} = B$ and $S_{T+d} = B_{T+d} = \emptyset$.
This enables us to conclude our induction and the proof for \eqref{eq:proof:balance}.

\end{proof}

\section{Proofs for random arrival order}
\label{app:ro_proofs}

 \begin{proof}[Proof of Proposition \ref{prop:cycle}]
 When $n$ is not a multiple of $p$, let  $v \in [1, p-1]$ be the remainder of the euclidian division of $n$ by $p$, and $u$ be such that  $n = pu + v$.
 Let $\{\sigma_1, ..., \sigma_K\}$ be a $p$-periodic $(\alpha, d)$-cover of $C_{n_1}^d$ with associated weights $\{\lambda_1, ..., \lambda_K\}$. We will show that it can be extended into an $(\alpha \left(\nicefrac{u}{u-2}\right), d)$-cover of $C_n^d$.

   We  set $\widetilde{\sigma}_{k}$ to be the $p$-periodic permutation over $1, ..., pu$ such that for all $i \in [1, p]$, $\widetilde{\sigma}_k(i) = \sigma_k(i)$.
  Let $x$ be an integer in the interval  $[1, u+1]$. Define the permutation $\sigma'_{k,x}$ as follows:
  \begin{equation}
  \begin{aligned}
  \sigma'_{k,x}(i)=
  \begin{cases}
  \widetilde{\sigma}_{k}(i) & i \leq px \\
   i + (u-x)p &  i \in [px+1, px+1 + v] \\
   \widetilde{\sigma}_{k}(i - v)  & i > px+1 + v.
  \end{cases}
  \end{aligned}
  \end{equation}

Take $i',j' \in [1, px] \cup [px+1+v, n]$ such that $|i' - j' | \leq d$.
  Because $n_1>p$ is a multiple of $p$, there exist $i, j \in [1, n_1]$ such that $i \equiv i' \mod p$, $j \equiv j' \mod p$ and $| i - j| \leq d $. By $p$-periodicity of $\sigma_k$ and $\sigma'_k$, we know that edge $(i',j')$ is in $B_n^{d}(\sigma'_k)$ iff $(i,j)$ is in $B_{n_1}^{d}(\sigma_k)$. Thus we can conclude that $\sum_k \lambda_k B_n^d(\sigma'_{k,x})$ covers edge $(i', j')$ of $C_n^d$.
  
  Every edge is therefore covered for at least $u-2$ different values of $x$. Therefore, $\sum_k \sum_x \frac{u}{u-2} \lambda_k B_n^d(\sigma'_{k,x})$ covers $C_n^d$. This means that $\left(\sigma_{k,x}\right)_{k, x}$ is an $\left(\alpha \left(\nicefrac{u}{u-2}\right), d\right)$-cover of $C_n^d$.

\end{proof}

\begin{proof}[Proof of Proposition \ref{prop:non_multiples} in case (ii)]
  Suppose now that $d + 1 = ku + v$ with $1 \leq v < k$. We first select vertices in the following way: select a subset $\Phi \subset [1, d+1]$ of $d + 1 - v$ vertices uniformly at random. Take $\Delta = \Phi + ku[1, r-1] = \{a + k u b | a \in \Phi, b \in [1, r-1] \}$ and note that $| \Delta | = kur$.
  
  We now contract vertices in $\Delta$. This is the same as in Definition \ref{def:contraction}: for $t \in [1, u]$, $a_t$ is the set of $u$ smallest vertices of $\Delta$ that are not in $a_1 \cup ... \cup a_{t-1}$. Because $d+1-v$ is a multiple of $u$, we have $a_{i+k}$ = $a_i + (d+1)$. This implies that the contracted graph is $C_{n/u}^{(d+1)/u}$.
  
  Similarly to the proof of case $(i)$, we extend a cover for $C_{n/u}^{(d+1)/u}$ to cover every edge $(i,j)$ for $i,j \in \Delta$. 
  If we sum over all the possible ways to select subset $\Phi$, we note that every edge $(i,j) \in C_{r(d+1)}^d$ is covered with probability at least $\left(\frac{d+1-v}{d+1}\right)^2$.
  \end{proof}

\section{Random arrival order: numerical values for small $d$}
\label{app:numerical}

Solving \eqref{eq:LP} and \eqref{eq:LP'} can become computationally difficult given the increase in the number of constraints, which is exponential in $d$.
In Table \ref{tab:small_d}, we show numerical values of solutions $\alpha_d$ and $\alpha'_d$ for $d$ between $2$ and $13$.
\begin{table}[ht!]
  \centering
\begin{tabular}{ |l ||l |l |l |l |l |l |l | l | l | l | l| l | l |}
  \hline
  $d$ & 1 & 2 & 3 & 4 & 5 & 6 & 7 & 8 & 9 & 10 & 11 & 12 & 13 \\ \hline
  $\alpha_d$ & 2 & 2.33 & 2.5 & 2.64 & 2.71 & 2.75 & 2.79 & 2.83 & $2.99^*$ & $3.2^*$ & $3.11^*$ & & $3.23^*$ \\ \hline
 $\alpha'_d$ &  & 4 & 3.45 & 3.17 & 3.15 & 3.12  & 3.09 & 3.08 & $3.07$ & $3.20^*$ & $3.153^*$ & $3.264^*$ & $3.318^*$ \\ \hline
\end{tabular}
  \caption{Numerical values for $\alpha_d$ and $\alpha'_d$ for small values of $d$. Starred elements were solved approximately (are therefore upper bounds on the actual value).}
  \label{tab:small_d}
\end{table}

We now need to provide upper bounds for all $d$ between $14$ and $51$. Note first that if $d+1$ is a multiple of $k$, then Proposition \ref{prop:non_multiples} implies that $\alpha_d \leq \alpha'_k$. Therefore, we need only consider prime values for $d$. In table \ref{tab:primes}, we compute upper bounds for $\alpha_d$ using Proposition \ref{prop:non_multiples}. For each case, we report which value of $k$ we used, as well as the value $\alpha'_k \left(\nicefrac{(d + 1 - v)}{(d+1)}\right)^2$. 

This allows us to conclude that Batching is $0.279$-competitive.

\begin{table}[ht!]

  \centering
\begin{tabular}{ |l ||l |l |l |l |l |l |l |l |l |}
  \hline
  $d$ & 17 & 19 & 23 & 29 & 31 & 37 & 41 & 43 & 47\\ \hline
  $\alpha_d \leq $ & 3.58 & 3.48 & 3.44 & 3.31 & 3.36 & 3.30 & 3.31 & 3.24 & 3.35\\ \hline
  $k$ used & 4 & 6 & 11 & 7 & 5 & 6 & 5 & 7 & 9 \\ \hline
\end{tabular}
  \caption{Upper bounds for $\alpha_d$ for prime values of $d$; derived from  Proposition \ref{prop:non_multiples} using the following formula: $\alpha_d \leq \alpha_k \left(\frac{d+1-v}{d+1}\right)^2 $ for $k \leq d$ and $v = d \text{ mod } k$ .}
  \label{tab:primes}
\end{table}

\end{document}